\documentclass[journal]{IEEEtran}

   \ifCLASSINFOpdf   
\else
\fi
\usepackage{times,epsfig,latexsym,amssymb,psfrag}
\usepackage{paralist,citesort}
\usepackage{booktabs}
\usepackage{subcaption}
\usepackage{eurosym}
\usepackage{mathrsfs}
\usepackage{enumitem}
\usepackage{graphicx}
\usepackage{epstopdf} 
\usepackage{amsmath}
\usepackage{mathtools}
 
 \newtheorem{theorem}{Theorem}[section]
 \newtheorem{lemma}[theorem]{Lemma}
 
 \newtheorem{corollary}[theorem]{Corollary}
 
 \newenvironment{proof}[1][Proof]{\begin{trivlist}
 		\item[\hskip \labelsep {\bfseries #1}]}{\end{trivlist}}

 \newenvironment{remark}[1][Remark]{\begin{trivlist}
 		\item[\hskip \labelsep {\bfseries #1}]}{\end{trivlist}}
 
 \newcommand{\qed}{\nobreak \ifvmode \relax \else
 	\ifdim\lastskip<1.5em \hskip-\lastskip
 	\hskip1.5em plus0em minus0.5em \fi \nobreak
 	\vrule height0.75em width0.5em depth0.25em\fi}
 
\DeclareMathOperator*{\mini}{minimize}
\DeclareMathOperator*{\maxi}{maximize}
\begin{document}
 \title{Optimized Resource Provisioning and Operation Control for Low-power Wide-area IoT Networks}

\author{Amin Azari and Meysam Masoudi and Cicek Cavdar\\
KTH Royal Institute of Technology, Email: 
\{aazari,masoudi,cavdar\}@kth.se} 
 \maketitle

\begin{abstract}
The tradeoff between cost of the access network and quality of offered service to IoT devices, in terms of reliability and durability of communications, is investigated. We first develop analytical tools for reliability evaluation in uplink-oriented large-scale IoT networks. These tools comprise modeling of interference from heterogeneous interfering sources with time-frequency asynchronous radio-resource usage patterns, and benefit from realistic distribution processes for modeling channel fading and locations of interfering sources. We further present a cost model for the access network as a function of provisioned resources like density of access points (APs), and a battery lifetime model for  IoT devices as a function of intrinsic parameters, e.g. level of stored energy, and network parameters, e.g. reliability of communication. The derived models represent the ways in which a required level of reliability can be  achieved by either sacrificing battery lifetime (durability), e.g. increasing number of replica transmissions, or  sacrificing network cost and increasing provisioned resources, e.g. density of the APs. Then, we investigate optimal resource provisioning and operation control strategies, where the former aims at finding the optimized investment in the access network based on the cost of each resource; while the latter aims at optimizing  data transmission strategies of IoT devices. The simulation results confirm tightness of derived analytical expressions, and show how the derived expressions can be used in finding optimal operation points of IoT networks.

\end{abstract}
\begin{IEEEkeywords}
Economic viability, Reliability and durability, Coexistence, Grant-free, IoT, LPWA.
\end{IEEEkeywords}
 
\IEEEpeerreviewmaketitle
 
 
 \section{Introduction}
Providing  connectivity for  massive Internet-of-Things (IoT) devices is a key deriver of 5G \cite{5g_iot}. Until now, several solutions have been proposed  for enabling large-scale  IoT connectivity, including evolutionary and revolutionary solutions \cite{mag_all}.  Evolutionary solutions aim at enhancing connectivity procedure of existing LTE networks, e.g.  access reservation and scheduling improvement \cite{isl,nL}. On the other hand, revolutionary solutions  aim at providing low-overhead scalable low-power IoT connectivity by redesigning the access network. In 3GPP LTE Rel. 13, narrowband IoT (NB-IoT) has been announced as a revolutionary solution which handles communications over a 200 KHz bandwidth \cite{ciot}. This narrow bandwidth brings high link budget, and offers extended coverage \cite{ciot}.  To provide autonomous low-latency access to radio resources, grant-free access is  a study item in 3GPP IoT working groups, and it is expected to be included in future 3GPP standards \cite{gf31}. 
Thanks to the simplified connectivity procedure, and removing need for pairing and fine synchronization, grant-free radio access has attracted lots of interests in recent years  for providing  low-power ultra-durable IoT connectivity, especially when more than 10 years lifetime is required.   SigFox and LoRa are two dominant grant-free radio access solutions over the ISM-band, the industrial, scientific, and medical radio band.   While energy consumptions of LoRa and SigFox solutions is extremely low, and their provided link budget is enough to penetrate to most indoor areas, e.g. LoRa signal can be decoded when it is  20 dB less than the noise level,  reliability of their communications in coexistence scenarios is questionable \cite{int2,mey}. Regarding the growing interest in grant-free radio access, it is required to investigate the  performance of grant-free IoT networks in terms of reliability/durability of communications, expected battery lifetime of devices, and the CAPEX and OPEX of the access network.
\subsection{Literature Study}
Non-orthogonal radio access  has attracted lots of attentions  in recent years as a  complementary radio access scheme for future generations of wireless networks \cite{noma,jsacS}. In literature, non-orthogonal access has been employed in order to increase the network throughput \cite{reem},  reliability \cite{url}, battery lifetime \cite{gf}, and reduce delay \cite{reem}.  A through survey of non-orthogonal radio access nominated for 5G can be found in \cite{cat}, which categorizes the available schemes into three categories: (i) codebook-based multiple access, with codebooks  in power or code domain like sparse code and pattern division multiple access (SCMA, PDMA); (ii)   sequence-based multiple-access, using   complex number sequences like multi-user shared access (MUSA) and non-orthogonal coded multiple access (NCMA); and (iii) interleaver/scrambler-based multiple access like resource spread multiple access (RSMA). Among these schemes, MUSA and RSMA can be used in an asynchronous and grant-free mode, and RSMA has been proposed as a candidate for  grant-free access in future LTE releases \cite{rsma}.
In \cite{miao2016MAC,adh}, grant-free non-orthogonal radio access for intra-group communication of IoT devices over cellular networks has been investigated, and  it has been shown that significant improvement in battery lifetime can be achieved with bounded  interference on communications of other cellular users. In low-power wide-area (LPWA) IoT technologies over unlicensed band, signal repetition in time and  spreading in frequency are mainly used  to combat noise and  interference while keeping device's cost and energy consumption as low as possible \cite{mag_all}.  In \cite{2d},  interference and outage probability in grant-free access has been investigated by assuming a constant received power from all contending devices, which is not the case in practice  regarding different pathloss values that different devices experience, and lack of channel state information as well as sophisticated power control at IoT-device side.  The success probability in grant-free transmission for a single cell has been analyzed in \cite{sic} by assuming a Poisson point process (PPP) distribution of  IoT devices. In \cite{mey}, the outage probability in grant-free transmission  has been analyzed by assuming Rayleigh fading and PPP distribution of IoT devices.  In \cite{gf}, a low-cost low-power grant-free radio access has been proposed, which benefits from oscillator imperfection of low-cost IoT devices for contention resolutions. Experimental performance evaluation results in \cite{int2} reflects a significant impact of interference from already installed ISM-band devices  on the performance of LPWA networks.

One sees the research on grant-free access has been mainly focused on success probability analysis in homogeneous scenarios. Furthermore, the choice of PPP for distribution of devices in LPWA IoT networks, where the cell range can be up to tens of kilometers \cite{mag_all}, leads to a loose upperbound \cite{adhoc,math,pcp}. This is   due to the fact that  there is a high density of IoT devices  in buildings, shopping centers, and etc., and a low density of nodes outside these regions. In this case,  a Poisson cluster process (PCP), which in special form reduces to PPP,  suits well modeling distribution of devices.  One sees that there is lack of research on reliability of large-scale IoT networks with multi-type devices with heterogeneous  communications characteristics and   distribution processes (Fig. \ref{fig1}).
Furthermore, there is lack of a unified approach investigating the tradeoff between economic viability, i.e. the required investment cost in the access network, and quality of service for IoT devices, in terms of  reliability and durability of communications.
\begin{figure*}[t!]
    \centering
    \begin{subfigure}[t]{0.5\textwidth}
   \centering
     \includegraphics[width=2.2in]{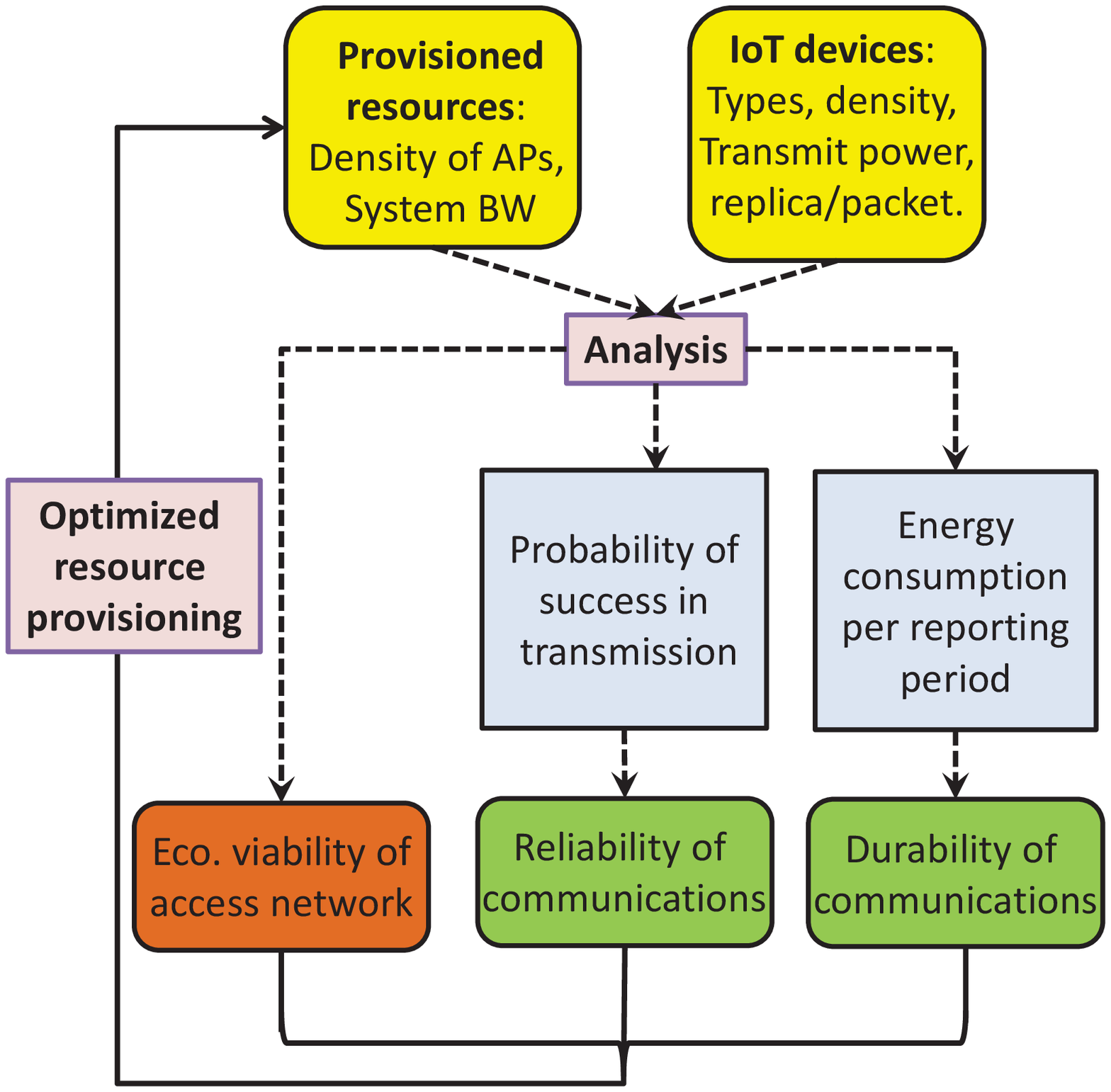}
\caption{}\label{trcn1}
    \end{subfigure}%
~
    \begin{subfigure}[t]{0.5\textwidth}
   \centering
     \includegraphics[width=2in]{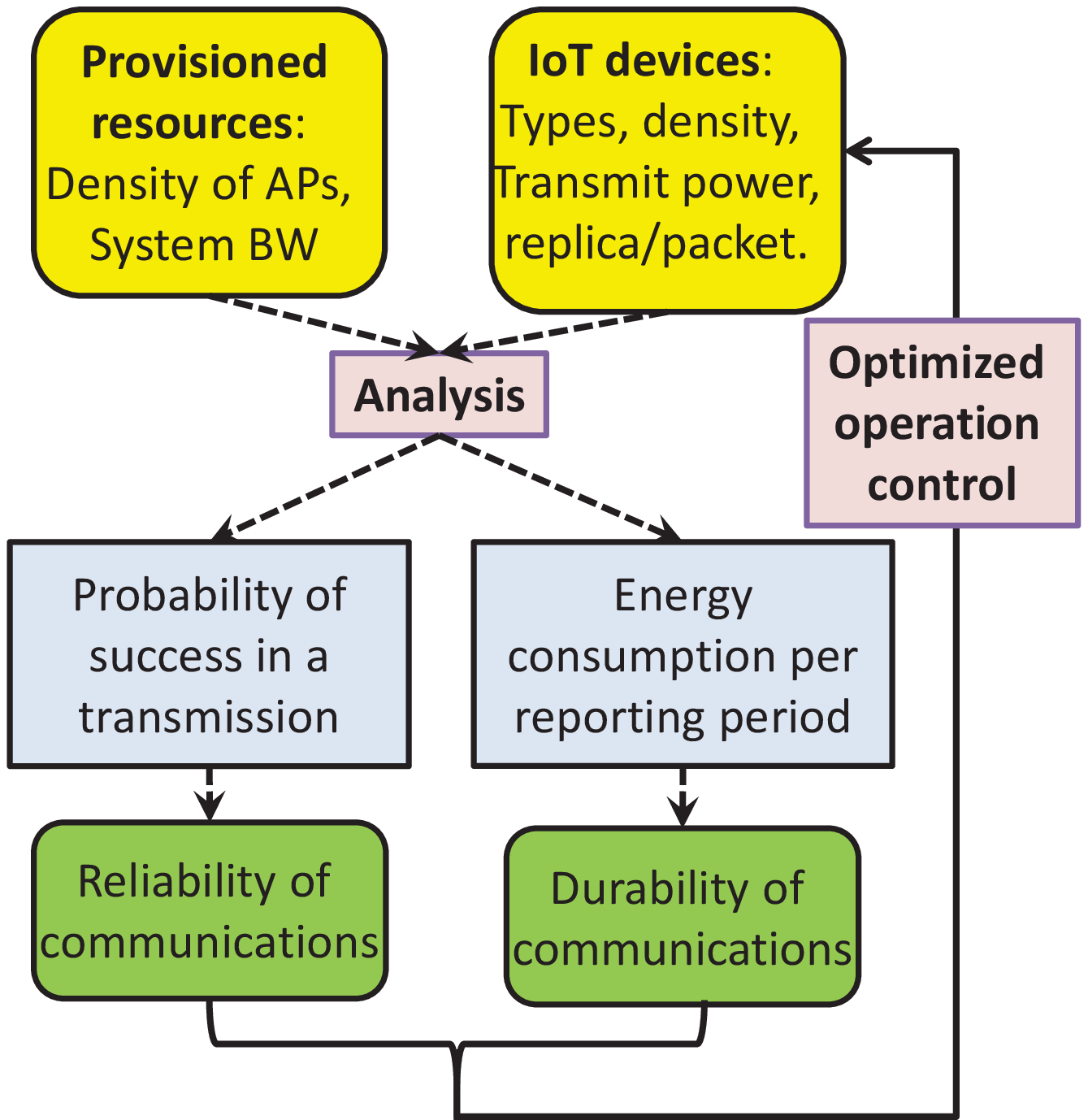}
\caption{}\label{trcn2}
    \end{subfigure}%
\caption{(a) optimized resource provisioning in the access networks for serving IoT traffic; (b) optimized operation control for IoT devices. Violet-colored boxes represent the contributions.}  \label{trcn}
\end{figure*}
 
\subsection{Motivation and Contributions}
  Here, we address an important  problem, not tackled previously: network design  in coexistence scenarios  with grant-free radio access. With coexistence, we  mean both an  IoT data aggregation solution which aggregates data from heterogeneous IoT traffic sources, e.g. a cellular IoT solution with different IoT services; and heterogeneous IoT solutions sharing a bunch of spectrum and operating together. For network design, the key performance indicators (KPIs) of interest are cost of the access network, and  reliability/durability of provided service for end-users. Enabling IoT connectivity requires deployment of access points (APs)  and allocation of frequency resources, which determine  the network costs. One the other hand, the experienced delay, consumed energy (battery), and success of IoT applications have strong couplings with reliability of data transfer, which is interconnected with the network resources. This tradeoff is investigated in this work. 
  
    The main contributions of this work include:
    \begin{itemize}
    \item
    Analytical modeling
    \begin{itemize}
    \item
Provide a rigorous analytical model of reliability in grant-free IoT connectivity  by considering large-scale IoT networks which serve multiple IoT traffic categories with heterogeneous communication characteristics and distribution processes.
     \item
    Provide  analytical model for durability of communications (lifetime of IoT application) as a function of expected battery lifetimes of devices.   Provide  analytical model for required investment cost for  the access networks as a function of provisioned resources.

     \item     
     Highlight the tradeoffs amongst the network investment  cost, durability, and reliability of communications.
     \end{itemize}
     \item
     Optimized design
     \begin{itemize}
     \item
Present reliability-constrained cost-optimized resource provisioning strategies for large-scale IoT networks. 
\item
Present reliability-constrained lifetime-optimized operation control strategies for IoT devices.
\end{itemize}
\item
Scalability analysis
\begin{itemize}
\item
Present scalability of the access network, i.e. how the provisioned radio and AP resources must be scaled with scaling of number of deployed devices, external interference, and required reliability level.
\item
Present scalability of IoT devices by investigating the point up to which IoT devices can adapt themselves to increase in co-usage of the shared medium by increasing the transmit power and number of replica transmissions.
\end{itemize}

\end{itemize}
Fig. \ref{trcn} represents a graphical illustration of contributions.

The remainder of paper has been organized as follows. System model and problem description are presented in the next section. Modeling of KPIs is presented in section III. Section IV presents the optimized resource provisioning and operation control strategies. Simulation results are presented in section V. Concluding remarks are given in section VI.



\begin{figure*}[t!]
    \centering
    \begin{subfigure}[t]{0.5\textwidth}
   \centering
     \includegraphics[width=2.5in]{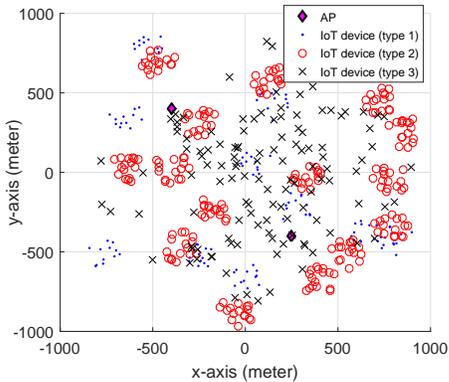}
\caption{Graphical illustration  of locations of devices  for $K=3$ which represents  heterogeneity in distribution processes of locations.
}\label{sys}
    \end{subfigure}%
~
    \begin{subfigure}[t]{0.5\textwidth}
   \centering
     \includegraphics[width=3in]{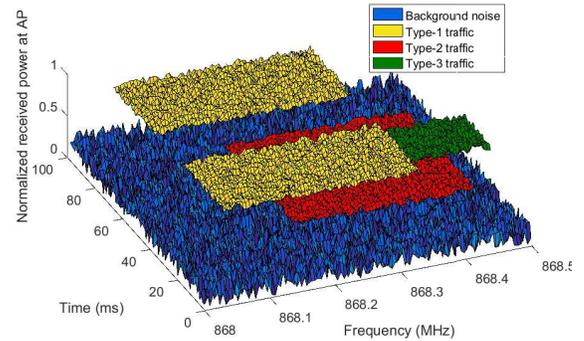}
\caption{A snapshot of received traffic from $4$ transmitters ($K=3$) at the receiver which represents  differences in  communication characteristics.}\label{crp}
    \end{subfigure}%
\caption{Graphical description of  the system model. Performance analysis and optimization with heterogeneity in communication characteristics as well as distribution processes of devices have not been  studied before.}  \label{fig1}
\end{figure*}

\section{System Model and Problem Description}
\subsection{System Model}\label{sys}
 A massive number of  IoT devices, denoted by set $\Phi$, have been distributed according to different spatial PCPs in a wide service area, as depicted in Fig. \ref{sys}. $\Phi$ comprises of $K$ subsets, $\Phi_k$ for $k\in \mathcal K \buildrel \Delta \over = \{1,\cdots, K\}$, where each subset refers to a specific type of IoT services. Traffic from different subsets differ in the way they use the time-frequency resources, i.e. in frequency of packet generation $1/T_k$, signal bandwidth $w_k$, packet transmission time $\tau_k$, number of replicas\footnote{Practical motivations for modeling such replicas can be found in newly proposed IoT technologies in which coverage extension and resilience to interference  are achieved by repetitions of transmitted packets \cite{ciot,mag_all}. When it is not the case, $n_k=1$ can be used.} transmitted per packet $n_k$, and transmit power $P_k$. Subscript $k$ refers to  the type of IoT devices.  For PCP of type-$k$ IoT traffic, $(\lambda_k, \upsilon_k, \text{f}({\bf x}))$ characterizes the process in which, $\lambda_k$ is the density of the parent process and $\upsilon_k$ the average number of daughter points per parent point, as defined in \cite{pcp}. Also, $\text{f}({\bf x})$ is an isotropic function representing  scattering density of the daughter process around a parent point, e.g. a normal distribution:
 \begin{equation}\label{nor}\text{f}({\bf x})={\exp(-||{\bf x}-{\bf x}_0||^2/(2\sigma^2))}/{\sqrt{2\pi\sigma^2}},\end{equation}
 where $\sigma$ is the variance of distribution and ${\bf x}_0$ is the location of parent point,     or a uniform distribution:
     $\text{f}({\bf x})=\text{S}(||{\bf x}-{\bf x}_0||)/(\pi R_c^2),$
      where $R_c$ is the cluster radius, $\text{S}(x)$ is 0 for $x>R_c$, and 1 otherwise.   A frequency spectrum of $W$ is shared for  communications, on which the power spectral density of noise is denoted by $\mathcal N$. We aim at collecting data from a subset of $\Phi$, denoted by $\phi$, where $|\phi|\le |\Phi|$, and $|\Phi|$ represents cardinality of $\Phi$. Devices in $\Phi$ also share a set of semi-orthogonal codes denoted by $\varpi $ with cardinality $C$, which reduces the interference from other devices reusing the same radio resource with a different code by factor of $\mathcal Q$. Examples of such codes are semi-orthogonal spreading codes in LoRa
technology \cite{mag_all}, and RSMA codes which are expected to be implemented in future releases of LTE for grant-free access \cite{rsma}. A list of frequently used symbols  has been presented in Table \ref{ab} for ease of reading.

\begin{table}[t!]
\centering \caption{Frequently used symbols}\label{ab}
\begin{tabular}{p{1.5 cm}p{6.5 cm}}\\
\toprule[0.5mm]
{\it Symbol }&{\it Description (Underscript $k$ refers to IoT type.)}\\
\midrule[0.5mm]
$\mathcal A$ & Service area\\
$\Psi_k$& Distribution process of locations of interfering devices  \\
 $\Psi$; $\mathcal K$ &$\cup_{k\in\mathcal K} \Psi_k$; Set of IoT types: $\{1,\cdots, K\}$\\
$\Phi_k$; $\Phi$& Set of type-$k$ devices; $\cup_{k\in\mathcal K} \Phi_k$\\
$\phi$; $|\phi|$& Subset of interest out of $\mathcal \Phi$; Cardinality of $\phi$\\
$\mathcal L_{I_\Psi}$& Laplace functional of interference from $\Psi$\\
${\bf x}, {\bf y}, {\bf z}$ & Points in the 2D plane\\
$\tau_k$& Packet transmission time\\
$T_k$ & reporting period\\
$\gamma_{\text{th}}$& Threshold SINR\\
$P_k$ & Transmit Power (Iot device)\\
$\lambda_k$& Density of parent points\\
$ \upsilon_k$& Avg. number of nodes per cluster\\
$\text{f}(x)$& Distribution function of daughter points of a parent point\\
$m,\Omega$ & Parameters of Nakagami-$m$ channels\\
$\text{g}(\cdot)$; $\delta$; $h$ & Pathloss function; Pathloss exponent; fading\\
$Q_j$ & Interference rejection factor\\
$W$ & System bandwidth\\
$\varpi $ & Number of shared (semi-)orthogonal codes\\
$\text{P}_{\text s}$; $\text{P}_{\text o}$ & Success probability; Outage probability\\
$N$, $\mathcal N$& Noise; Noise power\\
$n_k$& Number of transmitted replicas per packet\\
$L(k)$& Device-level battery lifetime; Application-level battery lifetime\\
$\mathbb L(k)$& Application-level battery lifetime\\
$\lambda_a$ &Density of APs\\
Notation& a: Symbol, $\{a,A\}$: parameter, $\textbf a$: vector\\
Notation& $\{\text{A}(\cdot$), $\mathcal{A}(\cdot)$, $\mathbb A(\cdot)\}$: Function\\
\bottomrule[0.5mm]
\end{tabular}
\end{table}

\subsection{KPIs' Description}\label{kpi}
\subsubsection*{Quality of Service (QoS) for IoT Communications}
In this work, we introduce two QoS measures  for IoT communications. The first  consists in probability of success in transmission of a packet within a bounded number of trials, or within a bounded time interval. Furthermore, regarding the fact that most IoT devices are battery-driven and long battery lifetime for them is of crucial importance in order to reduce the human intervention in battery replacement, we also introduce another measure of QoS which consists in expected battery lifetime of devices, or durability of communications. Assuring long battery lifetime for IoT devices, which can be seen as a long-term reliability, is the key for success of any large-scale IoT solution which aims at providing beyond 4G things-connectivity \cite{lif_com}.

\subsubsection*{Cost of the Access Network}
From the first to fourth generation of wireless networks (LTE), providing large-scale connectivity for low-cost IoT devices has not been a  main objective in design of  wireless infrastructure. Furthermore, the  communications characteristics and requirements of IoT traffic are fundamentally different  from the legacy traffic \cite{5g_iot}. These two facts have motivated researchers to think of revolutionary connectivity solutions for IoT traffic. At the design phase of any of such solutions, study of  technoeconomic models is of crucial  important, as those models shed light upon the cost which must be burden for providing a level of QoS in communications. These technoeconomic models as well as their interactions with QoS requirements and level of provisioned resources, are missing in most past done IoT research works.

\subsection{Problem Description} 
Grant-free time-frequency asynchronous radio access is a strong enabler of ultra-durable low-cost IoT connectivity \cite{gf}. When it comes to the grant-free operation, the main research question is  {\textit{how to design and adapt deployment and operation of an IoT network, respectively, with regard to the communication characteristics of services which are reusing the  radio resources}}. The former includes derivation of optimized density of the required access points and system bandwidth; while the latter includes derivation of transmit powers of devices and number of replica transmissions  per  data packet. In this  paper, we focus on derivation  of minimum cost network resources which are required such that a given level of IoT communications' reliability is maintained over a service area. 

\subsection{Applications} 
The derived results in this work can be used design and optimization of both IoT solutions over licensed-band, e.g. grant-free access in NB-IoT \cite{rsma}, and unlicensed-band, e.g. LoRa and SigFox \cite{mag_all}. This is due to the fact that the presented models and results capture heterogeneous external interference in reliability derivation, where such interference is an essential characteristic of unlicensed bands, as well the cost of the spectrum in the network costs calculations, where such cost is a characteristic of licensed bands. As mentioned in the system model, the aim is to  collect data from a subset of $\Phi$, denoted by $\phi$.  When $|\phi|=|\Phi|$, the analytical framework is applicable for grant-free access over licensed-band which doesn't suffer much from  external interference; while when $|\phi|<|\Phi|$, the framework is applicable for grant-free access over unlicensed band with coexisting technologies.

%
%

\section{Analytical Modeling of KPIs}
\subsection{Modeling of Reliability}
In our grant-free radio access system,  transmitting devices are asynchronous in time and frequency domains, and hence, the received packets at the receiver may have partial overlaps in time-frequency, as depicted in Fig. \ref{crp}.  To model reliability in communications,  we first derive an analytical  model for interference in subsection \ref{siI}, and for probability of success in subsection \ref{su1}. These models are then employed in deriving reliability of communications in subsection \ref{rels}.

\subsubsection{Interference Analysis}\label{siI}
We assume a type-$i$ device has been located at point $\bf z$ in a 2D plane, and its respective AP has been located at the origin.  In order to find probability of success in data transmission from the device to the AP, we need to characterize the received interfere at the AP.   A common practice in interference description is to determine its moments, which is possible by finding its generating function, i.e. the Laplace functional \cite{adhoc,math}.  Towards this end, let us introduce three stationary and isotropic processes: i) $ \Psi^{(1)} =\cup_{k\in \phi} \Psi_k^{(1)}$,  where $\Psi_{k}^{(1)}$ represents the PCP containing locations of type-$k$ transmitting nodes which are reusing radio resources with a similar code to the code\footnote{Note: as mentioned in the system model, devices in $\phi$ share a set of semi-orthogonal codes for partial interference management.} of transmitter of interest; ii) $ \Psi^{(2)} =\cup_{k\in \mathcal K}  \Psi_k^{(2)}$,  where $ \Psi_{k}^{(2)}$ represents the PCP containing locations of type-$k$ transmitting nodes which are reusing radio resources with a different code (or no code, in case $k\notin \phi$) than the transmitter of interest; and iii) $\Psi=\cup_{j\in\{1,2\}} \Psi_k^{(j)}$.
   For an AP located at the origin, the Laplace functional of the received interference at the receiver is given by:
 \begin{align}
 \mathcal L_{I_{\Psi}}(s)&=\mathbb E\big[\exp(-s I_{\Psi})\big]\label{base}\\
 &=\mathbb E\big[  \prod\nolimits_{j\in \{1,2\}} \prod\nolimits_{k\in \mathcal K}  \prod\nolimits_{{\bf x}\in \Psi_k^{(j)}}\mathcal L_h({sQ_j P_k  \text{g}({\bf x})})\big],\nonumber
  \end{align}
where $ Q_j P_k  \text{g}({\bf x})$ is the average received power  due to a type-$k$ transmitter  at point ${\bf x}$, $Q_1=1$, $Q_2=\mathcal Q$, and $\mathcal Q$ is the rate of rejection of interference between two devices with different multiple access codes, as defined in section \ref{sys}. Also,  $h$ is the power fading coefficient associated with the channel between the device and the AP, and $\mathcal L_{h}\big(s Q_j  P_k \text{g}({\bf x})\big)$ the Laplace functional of the received power. We consider the following general path-loss model 
$$ \text{g}({\bf x}) = 1/(\alpha_1 + \alpha_2||{\bf x}||^\delta),$$  where $\delta$ is the pathloss exponent, and $\alpha_1$ and $\alpha_2$ are control parameters.  When $h$ follows Nakagami-$m$ fading, with the shaping and spread parameters of $m\in \mathbb Z^i$ and $\Omega>0$ respectively, the PDF of the power fading coefficient is given by:
 \begin{equation}\text{p}_h(q)= \frac{1}{{\Gamma(m)}}(\frac{m}{\Omega})^m q^{m-1}\exp\big({-\frac{mq}{\Omega}}\big),\label{nm}\end{equation}
 where $\Gamma$ is the Gamma function.
 Then using Laplace  table, $\mathcal L_{h}\big(sQ_j  P_k \text{g}({\bf x})\big)$   is derived as:
\begin{equation}\label{laph}L_h(s  Q_j P_k \text{g}({\bf x}))={\big(1+{\Omega}s P_k \text{g}({\bf x})/m\big)^{-m}}.\end{equation} 
By inserting \eqref{laph} in \eqref{base} and considering the fact that the received interferences from different device subsets are independent, we have: 
 \begin{equation}
 \mathcal L_{I_{\Psi}}(s)=  \prod\nolimits_{j,k}\mathbb E_{{\bf x},{\bf y}}\bigg[  \prod\nolimits_{{\bf y}\in \Theta_{k}}\big(  \prod\nolimits_{{\bf x}\in \theta_{\bf y}^{(j)}}u({\bf x},{\bf y})\big)\bigg],\nonumber
  \end{equation}
  where  the set of parent points of type-$k$ is represented by $\Theta_{k}$,  and transmitting nodes which are daughter points of $y$ as $\theta_{\bf y}^{(j)}$.
  Also, $\mathbb E_x$ represents expectation over $x$,  
and  $$u({\bf x},{\bf y})={\big(1 {+}{\Omega}s Q_j P_k \text{g}({\bf x} {-}{\bf y})/{m}\big)^{-m}}.$$ 
The received interference over our packet of interest can be decomposed into two parts: i) interference from transmitters belonging to the cluster of transmitter, i.e.  daughter points of the same parent; and ii)  other  transmitters.
Let us denote the Laplace functional of interference from the former and latter nodes as $\mathcal L_{I_{\Psi}}^i(s)$ and $\mathcal L_{I_{\Psi}}^o(s)$ respectively. Then, we have: \begin{equation}\label{taj}\mathcal L_{I_{\Psi}}(s)=\mathcal L_{I_{\Psi}}^o(s)\mathcal L_{I_{\Psi}}^i(s).\end{equation}
To proceed further, we recall a useful lemma from \cite{borel,math}.

\begin{lemma}\label{lbo}
   If $  X \subset  \mathbb R^2$ is assumed to be a PPP with intensity function $\xi({\bf x})$,   for any Borel function $b$: $ \mathbb R^2\to [0,1]$, we have:
  \begin{equation}\label{bor}\mathbb E_{x} \big[  \prod\nolimits_{{\bf x}\in X}b({\bf x})\big]=\exp\big(-\int\nolimits_{\mathbb R^2}(1-b({\bf x}))\xi({\bf x})  d {\bf x}\big).\end{equation}
  \end{lemma}
  \begin{proof}
  Presented in \cite{borel}.
  \end{proof}
Using Lemma \ref{lbo}, and conditioning on $\Theta_{k}$ and $\theta_{\bf y}^{(j)}$,  one has:
 \begin{align}
 &\mathcal L_{I_{\Psi}}^o(s)\label{li}\\
&\buildrel (\text{a}) \over =  \prod\nolimits_{j,k}\mathbb E_{y}\bigg[  \prod\limits_{{\bf y}\in \Theta_{k}}\big\{\exp\big(\text{-}\hat \upsilon_{k,j}  \int\nolimits_{\mathbb R^2} [1\text{-}{u({\bf x},{\bf y})}]\text{f}({\bf x})d {\bf x}\big)\big\} \bigg]\nonumber,\\
 &\buildrel (\text{b}) \over =\exp\big( \text{-} {\textstyle \sum\limits_{j,k}}\lambda_k \int\limits_{\mathbb R^2}\big\{1\text{-}\exp\big(\text{-}\hat\upsilon_{k,j} \int\limits_{\mathbb R^2}[1\text{-} {u({\bf x},{\bf y})}]\text{f}({\bf x})d {\bf x} \big)\big\}d {\bf y}\big),\nonumber
 \end{align}
where Lemma \ref{lbo} has been used for the cluster of each parent point in  (a), and for the cluster consisting of parent points in (b). Also, in \eqref{li} the average numbers of interfering type-$k$ devices in each cluster for $j\in\{1,2\}$ are denoted as  $\hat \upsilon_{k,1}=\upsilon_k\frac{n_k\tau_k}{T_k} \frac{w_k}{W}\frac{1}{\varpi }$  and $\hat \upsilon_{k,2}=\upsilon_k\frac{n_k\tau_k}{T_k} \frac{w_k}{W}\frac{\varpi -1}{\varpi }$ for $k\in\phi$. In these two expressions, the first fraction represents the percentage of  time in which device is active, i.e. the time activity factor, the second fraction represents the ratio of bandwidth that device occupies in each transmission, i.e. the frequency activity-factor, and the third fraction represents the code-domain activity factor, i.e. the probability that two devices select the same code, i.e. $\frac{\varpi -1}{\varpi }$, or different codes ${\varpi -1}{\varpi }$. Then, for $k\notin \phi$, it is clear that $\hat \upsilon_{k,1}=0$, and $\hat \upsilon_{k,2}=\upsilon_k\frac{n_k\tau_k}{T_k} \frac{w_k}{W}$.

Following the same procedure  used for deriving $\mathcal L_{I_\Psi}^o(s)$, one can derive   $\mathcal L_{I_\Psi}^i(s)$  as:
\begin{align}
\mathcal L_{I_\Psi}^i(s)\text{=}&\prod\nolimits_{j\in\{1,2\}}\mathbb E_y \big[\mathbb E_x [  \prod\nolimits_{{\bf x}\in \theta_{\bf y}^{(j)}}u({\bf x},{\bf y}) ]\big]\label{lik}\\
\text{=}& \int\nolimits_{\mathbb R^2}\exp\big(\text{-}{\textstyle\sum_j}\hat \upsilon_{i,j}  \int\nolimits_{\mathbb R^2} \big(1\text{-}{u({\bf x},{\bf y})}\big)\text{f}({\bf x})d {\bf x}\bigg)\text{f}({\bf y}) d {\bf y}\nonumber.\end{align}

\subsubsection{Probability of Successful Transmission}\label{su1}
Let $N$ denote the additive
noise at the receiver. Using the above derived  interference model,  probability of success in  packet transmission of  a type-$i$ device, located at $\bf z$, to the AP, located at the origin, is derived as: 
\begin{align}
\text{p}_{\text{s}}(i,{\bf z})&=\text{Pr}({  P_ih \text{g}({\bf z})}\ge[{N+I_{\Psi}}] \gamma_{\text{th}})\label{suc}\\
&\buildrel (\text{c}) \over =   \sum\limits_{\nu=0}^{m\text{-}1}\frac{1}{{\nu}!}\int\nolimits_{0}^{\infty}\exp({-}\frac{\gamma_{\text{th}}m q}{\Omega   P_i \text{g}({\bf z})}) q^{\nu}  d \text{Pr}(I_\Psi\text{+}N\ge q)\nonumber\\
&\buildrel (\text{d}) \over =  \sum\nolimits_{{\nu}=0}^{m\text{-}1}\frac{(-1)^{\nu}}{{\nu}!}[\mathcal L_{I_{\Psi}}(s)\mathcal L_{N}(s)]^{({\nu})} \big|_{s=\frac{\gamma_{\text{th}}m}{\Omega   P_i \text{g}({\bf z})}}, \nonumber
\end{align}
where $[F(s)]^{({\nu})}=\frac{\partial^{\nu}}{\partial s^{\nu}}F(s)$, (c) follows from \cite[Appendix~C]{adhoc}  and \eqref{nm} in which $\text{p}_h(q)$ has been defined, and finally (d) follows from \cite[Lemma~3.1]{alm} and the fact that $\mathcal L (t^n \text{f}(t))=(-1)^n\frac{\partial^n}{\partial s^n}F(s)$.
Furthermore, $L_{I_{\Psi}}$ has been characterized in \eqref{li} and \eqref{lik}, and $\mathcal L_N(s)$ is the Laplace transform of noise and is characterized by knowing type of the noise, e.g. white noise. 

In order to get insights on how coexisting services affect each other, in the following we focus on  $m=1$, i.e. Rayleigh fading, and present a  closed-form approximation of the success probability. In section \ref{simsec}, we will show tightness of this expression.
\begin{theorem}\label{t1}
For $m=1$, success probability  in packet transmission  can be  approximated as:
\begin{align}
&\text{p}_{\text{s}}(i,{\bf z})\approx \text{P}_{{\text{\tiny N} }} \big[\exp\big(-\sum\limits_{j\in\{1,2\}} \sum\limits_{k\in\mathcal K}\lambda_k\hat \upsilon_{k,j} \text{H}({\bf z},1, \frac{Q_j P_k\gamma_{\text{th}}}{\Omega  P_i})\big)\big]\nonumber\\
&\hspace{1cm}\times\exp\big(-\sum\nolimits_{j\in\{1,2\}}{\hat\upsilon_{i,j}}  \text{H}({\bf z},\text{f}^*({\bf x}), \frac{Q_j\gamma_{\text{th}}}{\Omega })\big),\label{ps}
\end{align}
  where  $\text{f}^*(\cdot)=\text{conv}\big(\text{f}(\cdot),\text{f}(\cdot)\big)$, 
  \begin{align}
  \text{H}\big({\bf z},\text{f}^*({\bf x}), \xi)&=\int\limits_{x\in\mathbb R^2}\frac{\text{g}({\bf x})}{\text{g}({\bf  x})+\text{g}({\bf z})/\xi}\text{f}^*({\bf x})d {\bf x}\label{hf},\\
   \text{P}_{{\text{\tiny N} }}&=\exp\big(-\mathcal N\gamma_{\text{th}}/[\Omega  P_i \text{g}({\bf z})] \big)\label{pn},
  \end{align}
  and $\mathcal N$ is the noise power.
\end{theorem}
\begin{proof}
Appendix \ref{pt1}.
\end{proof}
  $\text{H}({\bf z},\text{f}^*({\bf x}),\xi)$ and $\text{H}({\bf z},1,\xi)$  could be derived in closed-form for most well-known pathloss and distribution functions, as follows.
   \begin{corollary}
For $\text{g}({\bf x})=\alpha||{\bf x}||^{-\delta}$, 
\begin{equation}\label{hfd} \text{H}({\bf z},1,\xi)=  ||{\bf z}||^2 \xi^{\frac{2}{\delta}} 2\pi^{2} \text{csc}({2\pi/\delta})/\delta.\end{equation}
\end{corollary}
\begin{proof}
By change of coordinates, ${\bf x}\to (r,\theta)$, we have:
\begin{align}
\text{H}\big({\bf z},1, \xi)&=\int\nolimits_{x\in\mathbb R^2}\frac{\alpha{||\bf x||}^{-\delta}}{\alpha{||\bf x||}^{-\delta}+\alpha{||\bf z||}^{-\delta}/\xi} d {\bf x}\nonumber\\
&=2\pi\int\nolimits_{0}^{\infty}\frac{1}{1+(r/||{\bf z}||)^\delta/\xi}  { r  d r}\nonumber
\end{align}
Solving this integral by using \cite[Eq.~3.352]{seri} or \cite[Corollary~3.2]{alm}, \eqref{hfd} is derived.
\qed
\end{proof}

 \begin{corollary}\label{cne}
For $\text{g}({\bf x})=\alpha||{\bf x}||^{-4}$,  and $\text{f}({\bf x})$ given in \eqref{nor},
\begin{align}
 \text{H}({\bf z},\text{f}^*({\bf x}),\xi)=&  \frac{||{\bf z}||^2  }{4 \sigma^2\sqrt\xi    }\bigg[\text{ci}(\frac{ ||{\bf z}||^2 }{4\sigma^2\sqrt{\xi}}  )\sin(\frac{||{\bf z}||^2 }{4\sigma^2\sqrt{\xi}}  )-\nonumber\\
&\hspace{1cm}\text{si}(\frac{||{\bf z}||^2 }{4\sigma^2\sqrt{\xi}}  )\cos( \frac{||{\bf z}||^2 }{4\sigma^2\sqrt{\xi}} )\bigg],\nonumber
\end{align}
where $\text{si}(\cdot)$ and $\text{ci}(\cdot)$ are well-known sine and cosine integrals, as follows:
$$\text{si}(x)=-\int\nolimits_{x}^{\infty}\frac{\text{sin} (t)}{t}dt,\hspace{2mm} \text{ci}(x)=-\int\nolimits_{x}^{\infty}\frac{\text{cos}(t)}{t}dt.$$ 
\end{corollary}
\begin{proof}
Appendix \ref{pcne}.
\end{proof}
\begin{remark}\label{r1}
Analysis of $\text{H}\big({\bf z},\text{f}^*({\bf x}), \xi)$ shows that it can be well approximated by $1$ for $\frac{\sqrt\xi||{\bf z}||^2   }{4  \sigma^2  }\gg1$. For theorem \ref{t1} in which $\xi=Q_j\gamma_{\text {th}}/{\Omega}$, $\big({\bf z},\text{f}^*({\bf x}), \xi)\approx 0$ for $j=1$ because $Q_1=\mathcal Q\approx 0$; and $\text{H}\big({\bf z},\text{f}^*({\bf x}), \xi)\approx 1$ for $j=2$ when $z\gg \frac{2\sigma \sqrt[4]\Omega}{\sqrt[4]{\gamma_{\text{th}}}} \buildrel \Delta \over =z_0$ because $Q_2=1$. 
\end{remark}
\begin{remark}
From theorem \ref{t1}, one sees that probability of success, $\text{p}_{\text s}(i,{\bf z})$, is a function of $||\bf z||$ rather than phase of $\bf z$.  Then, hereafter we use  $\text{p}(i,z)$ to denote  probability of success for communication distance of $z$.
\end{remark}

Theorem \ref{t1} offers a closed-form expression of success probability as a function of distance to the receiver, which is a powerful tool for performance analysis of  large-scale heterogeneous  IoT networks, as we will see in section \ref{simsec}.

Until now, we have derived the probability of success for a given communication distance to an AP. In most LPWA IoT networks, we  have APs with overlapping coverage areas, and there is no pre-established connection between device and the APs. Hence, in the following we investigate success probability in such scenarios where multiple APs might be able to decode a packet. 
 
Regarding the fact that  theorem \ref{t1} provides probability of success as a function of communication distance, given distribution process of APs, the expected communication distance to the neighboring APs, and hence,  probability of success in data transmission   can be derived. In PPP deployment of APs with density $\lambda_{\text{a}}$, the cumulative distribution function (CDF) of distance from a random point to the $\ell$th nearest AP, denoted by $d_{\ell}$ is given by:
\begin{equation}\text{P}_{d_{\ell}}(r)\text{=}1\text{-}\text{Pr}({\ell}\text{-1 APs in } \pi r^2)\text{=}1\text{-}\exp\big({-\lambda_{\text a} \pi r^2}\big)\frac{[\lambda_{\text a} \pi r^2]^{{\ell}-1}}{({\ell}-1)!}.\label{cdf}\end{equation}
Then, one can derive the average probability of success in packet transmission from a random point for type-$i$  as:
\begin{equation}\text{P}_\text{s}(i)=  1-\prod\nolimits_{{\ell}=1}^{\ell_{\max}} \int\nolimits_{0}^{\infty} \big(1-\text{p}_\text{s}(i,r)\big) ~d \text{P}_{d_{\ell}}(r).\label{cov}\end{equation}
 Also, $d \text{P}_{d_{\ell}}(r)$ is derived from \eqref{cdf} as follows \cite{dis}:
$$d \text{P}_{d_{\ell}}(r)=\exp(-\lambda_{\text a}\pi r^2)\frac{2(\lambda_{\text a}\pi r^2)^{\ell}}{r({\ell}-1)!}dr.$$

\begin{theorem}\label{t3}
For $\text{f}(x)$ given in \eqref{nor}, $\text{g}({\bf z})=\alpha||{\bf z}||^{-4}$,  we have:
 $$\text{P}_\text{s}(i)\approx  1-\prod\nolimits_{{\ell}=1}^{\ell_{\max}} \big[1-\frac{X_0}{\sqrt{{X_1}^{\ell-1}}} \exp(\frac{{X_2}^2}{4{X_1}^2} ) \mathcal G(X_3,\ell)\big],$$
  \begin{align}
&\text{where } X_0=\frac{(\lambda_{\text a}\pi)^\ell}{(\ell-1)!}\exp\big(-{\hat\upsilon_{i,2}} \big), X_1=\frac{ \mathcal N\gamma_{\text{th}}}{\Omega  P_i \alpha},\nonumber\\ 
  & X_2=\sum\limits_{j,k} \lambda_k\hat \upsilon_{k,j} 
  (\frac{\gamma_{\text{th}}Q_j P_k}{\Omega  P_i})^{0.5} \frac{\pi^{2}}{2} \text{csc}(\frac{\pi}{2})+\lambda_{\text a}\pi, X_3=\frac{X_2}{2\sqrt{X_1}}.\nonumber
   \end{align}
 Also,  $\mathcal G(X_3,\ell)=\int\nolimits_{\frac{{X_2}^2}{2X_1}}^{\infty}  (z\text{-}X_3)^{(\ell-1)}\exp(-z^2)dz,$ and could be derived for any $\ell$ in the form of error function, e.g. for $\ell_{\max}=2$: 
 \begin{align}
&\mathcal G(X_3,1)=-(\sqrt{\pi}(\text{erf}(X_3) - 1))/2,\nonumber\\
&\mathcal G(X_3,2)=\exp(-X_3^2)/2 + (X_3 \sqrt{\pi}(\text{erf}(X_3) - 1))/2.\nonumber
\end{align}
\end{theorem}
\begin{proof}
Appendix \ref{pt3}.
\end{proof}
 \subsubsection{Reliability of IoT Communication}\label{rels}
Now, we have the required tools to investigate reliability of IoT communications. Once a type-$i$ device has a packet to transmit, it transmits $n_k$ replicas of the packet, e.g. $n_i=1$, and listens for ACK from the AP(s). If No ACK is received in a bounded listening window, device retransmits the packet, and this procedure can be repeated up to $B_{i}$ times, where the bound may come from fair use of the shared medium \cite{int2,mag_all} or expiration of data.   If data transmission is unsuccessful in $B_i$ attempts, we call it an outage event.  The probability of outage for type $i$ in  such setting can be denoted as:
\begin{equation}\label{rel}\text{P}_{\text o}(i)=\big[1-\text{P}_{\text s}(i) \big]^{n_iB_i},\end{equation}
where $\text{P}_{\text s}(i)$ has been derived in theorem \ref{t3}.


 \subsection{Battery Lifetime Performance (Durability)}\label{bl}
\subsubsection{Device-level Battery Lifetime} 
Packet generation at each device for most reporting IoT applications can be seen as a Poisson process \cite{3g}. Then, one can  model energy consumption of a device as a semi-regenerative process where the regeneration point  has been  located at the end of each successful data transmission epoch \cite{nL}. For a given device of type-$i$, let us denote the distance to $\ell$th neighboring AP as $d_{\ell}$, the stored energy in batteries as $E_{0}$, static energy consumption per reporting period for data acquisition from environment and processing as $E_{\text{st}}$, circuit power consumption in transmission mode as $P_{c}$, and inverse of power amplifier efficiency as $\eta$. Then, the expected battery lifetime  is \cite{nL}: 
\begin{equation} 
L(i)= \frac{E_{0}}{{E_{\text{st}}+\beta_i E_\text{c}+   \beta_in_i (\eta P_{i}+ P_{\text c}) \tau_i}}T_i,\label{lif}
\end{equation}
where $E_\text{c}$  represents the average energy consumption in listening after each trial  for ACK reception, and $\beta_i$ represents the average number of  trials and is derived as:
$$\beta_i=\sum\limits_{j=1}^{B_i}j\bigg[1 {-}\big[\prod_{\ell}1-\text{p}_{\text{s}}(i,d_{\ell} )\big]^{n_i}\bigg]\big[\prod_{\ell}1 {-}\text{p}_{\text s}(i,d_{\ell} )\big]^{n_i[j-1]},$$
where $\text{p}_{\text s}(i,r)$ has been derived in theorem \ref{t1}.
\subsubsection{Applications-level Battery Lifetime} \label{albl}
The lifetime of a reporting IoT application can be defined as the length of time between the reference time and when application is considered to be nonfunctional. The instant at which an application becomes nonfunctional is dependent on  the correlation between gathered data by neighboring devices. In critical applications with sparse deployment of sensors, where losing even one node deteriorates the performance or coverage, the shortest individual battery lifetime (SIBL) may define the application lifetime. When correlation amongst gathered data by neighboring nodes is higher, the longest
individual battery lifetime (LIBL), or average individual battery lifetime (AIBL) might be defined as the application lifetime.
Using AIBL definition, the expected battery lifetime for IoT application of type $k$ can be   approximated as $\mathbb L(i)\approx L(i)\big|_{\beta_i=\hat\beta_i}$, 
where:
$$\hat \beta_i=\sum\nolimits_{j=1}^{B_i}j\big[1\text{-}[1\text{-}\text{P}_{\text{s}}(i )]^{n_i}\big]\big[1\text{-}\text{P}_{\text s}(i )\big]^{n_i[j-1]},$$
and $L(i)$ and $\text{P}_{\text s}(i)$ have been derived \eqref{lif}, and  theorem \ref{t3} respectively.

 \subsection{Access Network's Cost}
The  access network's cost can be modeled as sum of spectrum, infrastructure, and operation costs.  Then,  the total annual cost  in a service area of $\mathcal A$, is derived as \cite{sib}:
\begin{align}
C_{\text{tot}} {=}c_1 \lambda_{\text a} \mathcal A {+}&c_2\lambda_\text{a} \mathcal A E_{\text{cons}} {+} {c_3}{W},
\label{cost}\end{align}
where  $c_1$ [\euro/AP] is the annual cost per AP  excluding the energy cost,  $c_2$ [\euro/Joule] is the annual cost for energy, and $c_3$[\euro/Hz] is the annualized spectrum cost, which is zero for unlicensed band. Also, $E_{\text{cons}}$ is the energy consumption per unit time, and depending on the type of the LPWA network, could be modeled in different ways. A proposed modeling for AP energy consumption per unit time is as
$E_{\text{cons}}=P_{\text r}+ P_{\text a}\sum\nolimits_{k\in\phi}\Lambda_k $, in which  $P_{\text r}$ is the load-independent power consumption, e.g. in listening to the channel and processing, $P_{\text a}$ the load-dependent power consumption  in forwarding received data to the core network and responding to the sender if required, and $\Lambda_k$ is the arrival rate of packets of type-$k$  devices deployed in the coverage area of an AP. When coverage areas of different APs are not overlapping, $\Lambda_k$ can be described as $\lambda_k\upsilon_k /[\lambda_\text{a} T_k]$. 

\section{Optimized Resource Provisioning  and Operation Control Strategies}

 \subsection{Tradeoff Analysis} \label{trf}
From the reliability expression in \eqref{rel} and theorem \ref{t3}, one sees that probability of success for type-$i$ traffic increases in (i) increase in the density of the APs, i.e. $\lambda_{\text a}$, which reduces the average communication distance; (ii) increase in bandwidth of communication, i.e. $W$, which reduces the probability of collision;  (iii) increase in transmit power of type-$i$ devices, i.e. $P_i$, and (iv) increase in number of replicas per packet, i.e. $n_i$. 
Let first focus of the first two items. \eqref{cost} shows that network cost increases with increase in the density of the APs and bandwidth of communication (in the case of licensed spectrum). Then, there is a clear tradeoff between reliability of provided communications and the investment cost. Furthermore, regarding the fact that the impacts of bandwidth and AP density on the reliability of communications are not the same, there should exist an optimal investment strategy to minimize the cost while complying with  the reliability constraints, as we will  show in section \ref{simsec}. 

Now we focus on the third and fourth aforementioned items for increasing  reliability of communications. From the battery lifetime analysis in \eqref{lif}, one sees that lifetime of devices may decrease in $n_i$ and $P_i$ because of the potential increase in the energy consumption per reporting period.   Furthermore, when reliability of communication is lower than a certain level, increase in $n_i$ and $ {P}_i$ may decrease the need for listening to the channel for ACK arrival and retransmissions, and hence, increasing $n_i$ or $ {P}_i$ may increase the battery lifetime. Taking this into account, one sees there should be an operation point beyond which, increase in $ {P}_i$ and $n_i$ offers a tradeoff between reliability and lifetime, and before which, it increases both reliability and durability of communication. This observation will be confirmed by simulation results in section \ref{simsec}.    From the above discussion, one sees that the design objectives, i.e. network cost, battery lifetime, and reliability of communications, cannot be treated separately in resource provisioning and operation control problems because they are coupled in conflicting ways such that improvements in one objective may lead to deterioration of the others. 

\subsection{Optimized Resource Provisioning}\label{orp}
As mentioned above, increasing $W$ and $\lambda_{\text a}$ have different impacts on reliability of communications as well as they result in different cost levels for the access network. Then, an interesting research problem consists in deriving the  optimized amount of investment in densification and   spectrum leasing, i.e.:
\begin{align}
\mini_{{\lambda_{\text a}, W}}&~C_{\text{tot}}  \buildrel \Delta \over =  [c_1  \mathcal A {+}c_2 \mathcal A E_{\text{cons}}]\lambda_{\text a} {+} {c_3}{W}\label{op1}\\
&\text{s.t.:} ~\text{P}_{\text o}(i)\le \text{P}_{\text o}^{\text{req}}(i), \forall i\in \phi. \nonumber
\end{align}
where $\text{P}_{\text{s}}^{\text{req}}(i)$ is the required reliability level.  Then, from the constraint is \eqref{op1}, we have:
\begin{align}
\text{P}_{\text o}(i)=\big[1-\text{P}_{\text s}(i) \big]^{n_iB_i}&\le {\text{P}_{\text o}^{\text{req}}}(i) \nonumber\\
1-\text{P}_{\text s}(i) &\le \sqrt[n_iB_i]{\text{P}_{\text o}^{\text{req}}(i)}
\nonumber\\
{\text{P}_{\text s}^{\text{req}}}(i) \buildrel \Delta \over = 1-\sqrt[n_iB_i]{\text{P}_{\text o}^{\text{req}}}(i)&\le \text{P}_{\text s}(i)
\label{cons}.
\end{align}
Now, by using the $\text{P}_{\text s}(i)$ expression in theorem \ref{t3} with\footnote{Following the same procedure, the results can be derived from any $\delta$. } $\delta=2$, and satisfying the  constraint in \eqref{cons} by equality, we have:
\begin{align}
\text{P}_{\text s}^{\text{req}}(i)&= \int\nolimits_{0}^{\infty} X_0 \exp(\text{-}X_5r^{2})2rdr\nonumber\\
&=\frac{0.5\sqrt{\pi}{\lambda_{\text a}\pi}\exp\big(-{\hat\upsilon_{i,2}} \big)}{\sum_{k}\lambda_{k}\hat \upsilon_{k,2} 
  (\frac{P_k\gamma_{\text{th}}}{P_i\Omega   })^{0.5} \frac{\pi^{2}}{2} \text{csc}(\frac{\pi}{2})\text{+}\lambda_{\text a}\pi\text{+}\frac{ \mathcal N\gamma_{\text{th}}}{\Omega  P_i \alpha}},\label{den}
\end{align}
in which, $\ell_{\max}=1$ and  $\mathcal Q\approx1$ have been assumed for brevity of expressions. Also, 
  $X_5$ is an auxiliary variable equal to the denominator of \eqref{den}.
This expression can be rewritten as:
\begin{align}
{\frac{0.5\sqrt{\pi}{\lambda_{\text a} }}{\text{P}_{\text s}^{\text{req}}(i)}\exp(-\frac{A_1}{W})}={\frac{A_2}{W}+\lambda_{\text a}+\frac{ \mathcal N\gamma_{\text{th}}}{\pi\Omega  P_i \alpha}},\nonumber
\end{align}
where:
$$A_1={\upsilon_{i}\frac{n_i\tau_i}{T_i} \frac{\varpi -1}{\varpi } }w_i;$$
$$ A_2=\sum_{k}{\upsilon_{k}\frac{n_k\tau_i}{T_k} \frac{\varpi -1}{\varpi } }w_k\lambda_{\text i}(\frac{P_k\gamma_{\text{th}}}{P_i\Omega   })^{0.5} \frac{\pi}{2} \text{csc}(\frac{\pi}{2}).$$
Solving this equation for $\lambda_{\text a}$, we derive an important expression for bandwidth-AP density tradeoff, as follows:
\begin{equation}\label{eqt}\lambda_{\text a}=\frac{\frac{ \mathcal N\gamma_{\text{th}}}{\pi\Omega  P_i \alpha}+\frac{A_2}{W}}{ \frac{0.5\sqrt{\pi}}{{\text{P}_{\text s}}^{\text{req}}} \exp(\frac{A_1}{W})-1}.\end{equation}
This expression shed light on the interconnection between the required bandwidth and AP density in providing a required level of reliability for IoT communications.  Using \eqref{eqt}, the optimization problem in \eqref{op1} reduces to a simple search over $W\ge 
{\log(\frac{0.5\sqrt{\pi}}{ {\text{P}_{\text s}}^{\text{req}}})}/{A_1}$ for minimizing:
\begin{align}
 [c_1  \mathcal A {+}c_2 \mathcal A E_{\text{cons}}]\max_{i\in\phi}\bigg\{\frac{\frac{ \gamma_{\text{th}}\mathcal N}{\pi\Omega  P_i \alpha}+\frac{A_2}{W}}{\frac{0.5\sqrt{\pi}}{   {\text{P}_{\text s}}^{\text{req}}(i)  \exp(\frac{A_1}{W})}\text{-}1} \bigg\}{+} {c_3}{W}.\label{reop}
\end{align}
 In \eqref{reop}, the multiplicand of $C_0$ represents the AP density based on the QoS requirement of the  most critical IoT-type.  Fig. \ref{costtr} represents the  tradeoff presented in \eqref{eqt} as well as the objective cost function in \eqref{reop} to be minimized, which is a quasi-convex function in this case. One sees as the required density of of APs decreases by increase in the system bandwidth, the network cost decreases to some point, and beyond which, network cost increases in system bandwidth. The point at which change in behavior of cost function occurs is a function of cost factors, i.e. $c_1, c_2, c_3$, and moves towards right-side in Fig. \ref{costtr} by decrease in $\frac{c_3}{c_1}$.

   \begin{figure}[t!]
 	\centering
 	\includegraphics[width=3.5in]{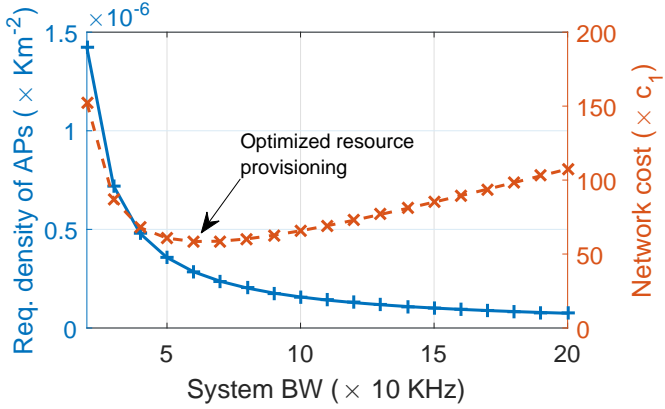}
 	\caption{The optimized resource provisioning problem. Required $ \text{p}_{\text s}(1,d_{\text{eg}})=0.5$, $\lambda_1$=1.6,  and other parameters can be found in Table \ref{sim}. }
 	\label{costtr}
 \end{figure} 
\subsection{Optimized Operation Control}
Increasing number of replicas per message as well as the transmission power can increase the reliability of communications, as discussed in subsection \ref{rels}. On the other hand, they may also increase or decrease the battery lifetime, as discussed in subsection \ref{trf}. Furthermore, they can affect reliability and battery lifetime performance of other IoT types due to the fact the received interference by other devices increase as transmit power or number of transactions of one type increase. Then, finding the optimized operation points is of paramount importance, as presented in the following. 
 Using the application-level battery lifetime definition in subsection \eqref{albl}, one may define the optimization problem for deriving the optimized operating points as follows:

\begin{align}
\maxi_{n_i, P_i, i\in\phi} & \hspace{3mm}\mathbb L(i); \label{op2}\\
&\text{s.t.:} ~\text{P}_{\text o}(i)\le \text{P}_{\text o}^{\text{req}}(i), n_i\le n_{max}, P_i\le P_{\max} \nonumber
\end{align}
Using \eqref{cons}, the reliability constraint can be rewritten as:
\begin{equation}\label{con}1-\sqrt[n_iB_i]{\text{P}_{\text o}^{\text{req}}}(i)\le \text{P}_{\text s}(i),\end{equation}
 where $\text{P}_{\text{s}}(i)$ has been derived in \eqref{den} as:
\begin{equation}\text{P}_{\text{s}}(i)=\frac{D_0}{\frac{1}{\sqrt{P_i}}D_1+\lambda_{\text a}\pi+\frac{ \mathcal N\gamma_{\text{th}}}{P_i\Omega   \alpha}},\label{rpi}\end{equation}
and the auxiliary variables $D_0$ and $D_1$ are defined as:
$$D_0=0.5\sqrt{\pi}{\lambda_{\text a}\pi}\exp\big(-{\hat\upsilon_{i,2}} \big),$$
 $$  D_1=\sum\nolimits_{k}\lambda_{k}\hat \upsilon_{k,2} 
  (\frac{P_k\gamma_{\text{th}}}{\Omega   })^{0.5} \frac{\pi^{2}}{2} \text{csc}(\frac{\pi}{2}).$$
Satisfying  \eqref{con} with equality, we have:
 $$\sqrt[n_iB_i]{\text{P}_{\text o}^{\text{req}}}(i)=1- \frac{D_0}{\frac{1}{\sqrt{P_i}}D_1+\lambda_{\text a}\pi+\frac{ \mathcal N\gamma_{\text{th}}}{P_i\Omega   \alpha}}.$$
 By simplifying the expression, $n_i$ is derived as a function of $B_i$ as follows:
\begin{equation} {n_i}=\left \lceil {\log(\sqrt[B_i]{\text{P}_{\text o}^{\text{req}}})}\bigg/{\log(1- \frac{D_0}{\frac{1}{\sqrt{P_i}}D_1+\lambda_{\text a}\pi+\frac{ \mathcal N\gamma_{\text{th}}}{P_i\Omega   \alpha}})}\right \rceil.\label{rni}\end{equation}
Also, the constraint on $n_i$ is translated to a constraint on $P_i$ as:
$$P_i\ge P_{min} \buildrel \Delta \over =\big(\frac{-{D_1}\text{+}\sqrt{{D_1}^2\text{-}4\frac{\mathcal N\gamma_{\text{th}}}{\Omega \pi}(\lambda_{\text a}\pi\text{-}\frac{D_0}{1\text{-}\sqrt[n_{\max}B_i]{\text{P}_{\text o}^{\text{req}}}})}}{2(\lambda_{\text a}\pi\text{-}\frac{D_0}{1-\sqrt[n_{\max}B_i]{\text{P}_{\text o}^{\text {req}}}})}\bigg)^2.$$
 Then, the optimization problem in \eqref{op2} reduces to a simple search over $ P_{min} \le \mathcal{P}_i\le P_{\max} $ for minimization of\footnote{Minimization of the expression in \eqref{ecprp} is equivalent to maximization of the battery lifetime expression in \eqref{lif}.}:
 \begin{equation} \label{ecprp}
 {\hat \beta_i E_\text{c}+   \hat \beta_in_i (\eta P_{i}+ P_{\text c}) \tau_i},
\end{equation}
in which $n_i$ has  been found as a function of $P_i$ in \eqref{rni},   
$$\hat \beta_i=\sum\nolimits_{j=1}^{B_i}j\big[1\text{-}[1\text{-}\text{P}_{\text{s}}(i )]^{n_i}\big]\big[1\text{-}\text{P}_{\text s}(i )\big]^{n_i[j-1]},$$
and $\text{P}_{\text s} (i)$ has been found as a function of $P_i$ in \eqref{rpi}.  Example of optimized operation control can be seen in Fig. \ref{oo}. In the next section, we further investigate the operation control and resource provisioning optimization problems, and   provide numerical results to show usefulness of the derived expressions in IoT network planning and optimization.


\begin{table}[t!]
\centering \caption{Simulation Parameters   }\label{sim}
\begin{tabular}{p{3.5 cm}p{4.6 cm}}\\
\toprule[0.5mm]
{\it Parameters }&{\it Value}\\
\midrule[0.5mm]
Service area &  $20\times 20 \text{ Km}^2$\\
Pathloss & $133+38.3\log (\frac{x}{1000})$\\
Thermal noise power & $-174$ dBm/Hz \\
Distribution process of devices & PCP$\big(\lambda_i\times$1e-6,$200$, Eq. \eqref{nor} with $\sigma$=100)\\ 
Packet arrival of each device & Poisson distributed with average reporting period ($T_i$) of 300 s\\
Packet transmission time ($\tau_i$) & 100 ms\\
Signal BW& 10 KHz\\
$E_0,P_{\text c}, \eta, E_{\text{st}}=0.5 E_{\text c}$& 1000 J, 10 mW, 0.5, 126 mW, 0.1 J\\
$P_{\text r}$, $P_{\text a}$ & 0.5 W, 1.5 W\\
$\gamma_{\text{th}}$, $\varpi$, $\eta$&1,1,0.5\\
$P_i, n_i,\lambda_{a}, W$ &  Default: 21 dBm, 1, 5.5{\rm e}-8, 100 KHz  \\
$\ell_{\max},\mathcal Q$&1, 0\\
$c_2,c_3$&$c_1/2000$, $c_1/2270$\\
\bottomrule[0.5mm]
\end{tabular}
\end{table}

 \section{Performance Evaluation}\label{simsec}

In this section we present the performance evaluation results. Towards this end, we implement a simulator for a large-scale IoT network with $K$-type IoT devices in MATLAB. Different IoT types differ in distribution processes of locations of devices, and communication characteristics. For type $i$, the distribution process of locations of respective devices is characterized by a PCP with density of parent points $\lambda_i$ (in Km$^{-2}$), $\upsilon_i=200$, and $f(x)$ given in \eqref{nor} with $\sigma$=100. The reliability constraint is described as $\text{p}_{\text s}(i,d_{\text{eg}})$, where $d_{\text{eg}}=\sqrt{{1}/{\pi\lambda_{\text a}}}$ is equivalent to the cell-edge communications distance in case of grid deployment of APs. The packet arrival at each node follows a Poisson process with rate $\frac{1}{T_i}$.   The default value of remaining  parameters can be found in Table \ref{sim}.

 \subsection{Validation of Derived Analytical Expressions}
 First, we investigate tightness of derived analytical expressions. By considering an IoT network comprising of two IoT types with different distributions and transmit powers, Fig. \ref{val} represents probability of success in packet transmission  for type-1 as a function of distance from the AP. One sees that the analytical model matches well with the simulation results. We have further depicted the contributions of noise, interferences from the same and other clusters of type-1 devices\footnote{Recall from \eqref{taj} in which we decomposed the received interference.}, as well as interference from type-2 devices. Regarding the fact transmit power of type-2 devices is 4 dB higher than type-1 devices, it is clear that  interference from type-2 traffic (plus-marked curve) is the most limiting factor. 
 
 \begin{figure}[t!]
 	\centering
 	\includegraphics[width=3.5in]{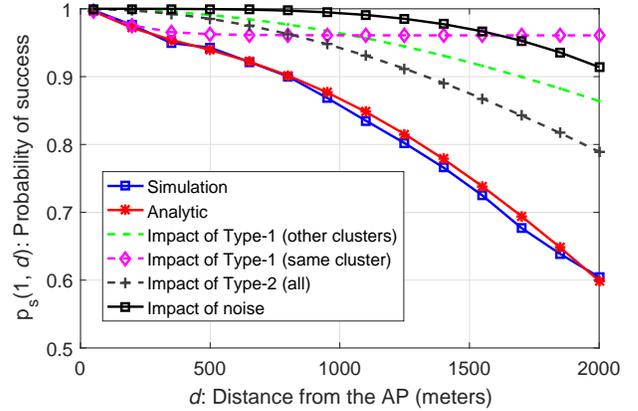}
 	\caption{Validation of analytical and simulation results. Device distribution:  $K$=$2$, $\lambda_1$=0.19, $ \lambda_2$=3.8, $\upsilon_1$=1200, $\upsilon_2$=30, $P_1$=21 dBm, and $P_2$=25 dBm. }
 	\label{val}
 \end{figure}

  \subsection{Analysis of Performance Tradeoffs}
  
   \begin{figure}[t!]
 	\centering
 	\includegraphics[width=3.5in]{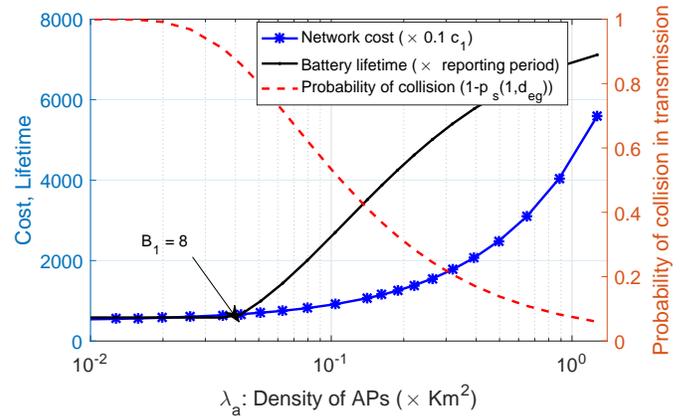}
 	\caption{Tradeoff between network cost, reliability, and battery lifetime ($K=1, \lambda_1$=6.4). For $\lambda_a\le 0.4$ Km$^2$, the maximum number of retransmissions, i.e. $B_1=8$, is met, i.e. most packets are dropped.   }
 	\label{pt}
 \end{figure} 
 
 Fig. \ref{pt} represents the tardeoff between cost of the access network (left $y$-axis), durability of communications (left $y$-axis), and reliability of communications.  The $x$-axis represents the density of APs. One sees by increase in density of APs, probability of collision in packet transmission decreases which results in  less required number of retransmissions. Then, it is clear that battery lifetime can be significantly saved by provisioning more resources for IoT traffic. This on the other hand increases network cost. One sees that the objectives of networks design are coupled such that improve  in one objective deteriorate the other. Hence, it is of crucial importance to provision network resources based on the actual need for QoS of IoT communications.
Further results on performance tradeoffs can be seen in subsection \ref{orp}. 
 
  \subsection{Optimized Deployment  Strategies}
 
 \begin{figure}[t!]
    \centering
    \begin{subfigure}[t]{0.5\textwidth}
   \centering
 	\includegraphics[width=3.5in]{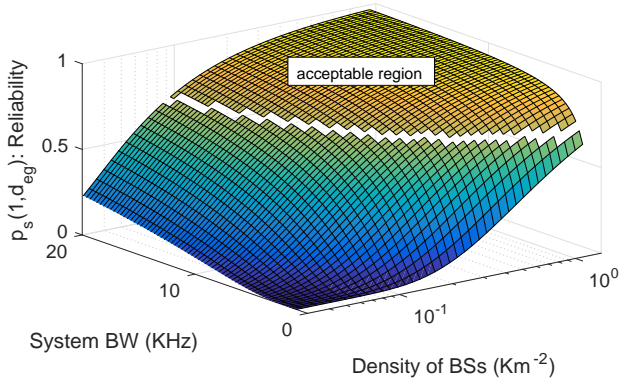}
 	\caption{Specifying region in which reliability constraint is satisfied.}
 	\label{od1}
    \end{subfigure}%
\\
    \begin{subfigure}[t]{0.5\textwidth}
   \centering
 	\includegraphics[width=3.5in]{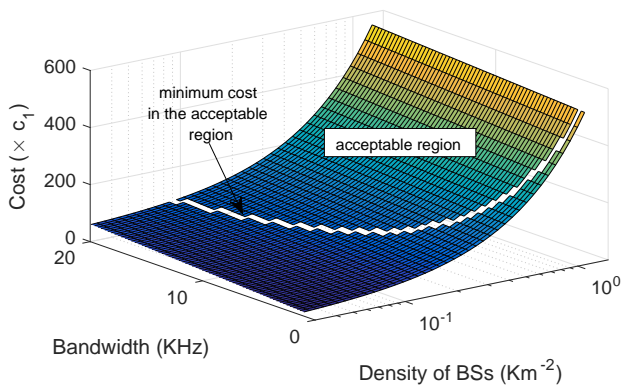}
 	\caption{Finding minimum-cost investment strategy  in the acceptable region.}
 	\label{od2}
    \end{subfigure}%
\caption{Optimized deployment strategy ($\lambda_1$=4.6,  required $p_{\text s}(1,d_{\text{eg}})$=0.7).    }  \label{od}
\end{figure}

Fig. \ref{od} represents the resource provisioning problem which has been partially touched in subsection \ref{orp}. First, Fig. \ref{od1} illustrates reliability of communications for different AP density-system BW configurations. One sees by increase in both AP density and system BW, probability of success in communications increases. Taking 0.7 as the required success probability, the density-bandwidth region in which reliability constraint is satisfied has been depicted in Fig. \ref{od1}. Using these results, in Fig. \ref{od2} we have depicted network cost as a function of provisioned bandwidth-density resources. Now, it is straightforward to search over the acceptable region to find the  bandwidth-density pair over which, network cost is minimized, as depicted in Fig. \ref{od2}.

  \subsubsection{Optimized Operation  Control} \label{sooc}

\begin{figure}[t!]
    \centering
    \begin{subfigure}[t]{0.5\textwidth}
   \centering
  \includegraphics[width=3.5in]{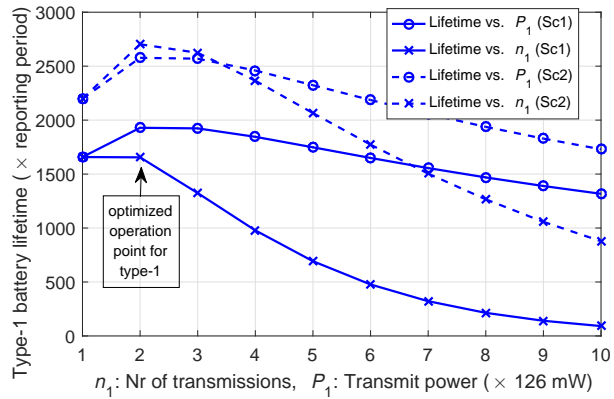}
\caption{Battery lifetime for type-1}\label{oo2}
    \end{subfigure}%
\\
    \begin{subfigure}[t]{0.5\textwidth}
   \centering
    \includegraphics[trim={0.35cm 0.00cm 0cm 0cm},clip,width=3.5in]{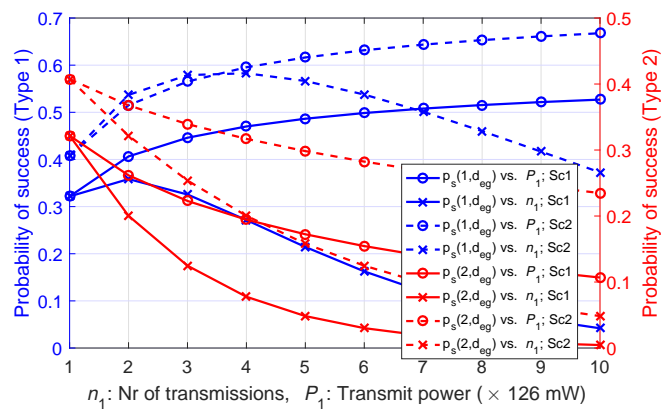}
\caption{Probability of success in transmission for  type-1 and type-2 devices}\label{oo1}
    \end{subfigure}%
\caption{Optimized operation control ($K=2, \lambda_2$=2.4, $\lambda_1$=2.4 in Sc1 and $\lambda_1$=1.2 in Sc2). $n_1$ and $P_1$ refer to the number of replica transmissions per data packet, and  transmit power respectively.  In circle-marked curves, $n_1=1$ and $P_1$ is varying. In  plus-marked curves, $P_1=126$ mW and $n_1$ is varying. Other parameters are presented in Table \ref{sim}. }  \label{oo}
\end{figure}

 Fig. \ref{oo} represents the  interactions amongst $n_i$, $P_i$, success probability, and battery lifetime. The $x$-axis in Fig. \ref{oo2} and Fig. \ref{oo1} represents  $P_i$ for circle-marked curves and $n_i$ for cross-marked curves. Also, Sc2 refer to a similar setting to Sc1 in which density of type-1 nodes has been reduced by a factor of 2. One sees in Fig. \ref{oo2} that battery lifetime is a quasi-concave function of both $P_i$ and $n_i$. Also, in Sc1, where density of nodes is higher than Sc2, battery lifetime decreases significantly by increase in number of replica transmissions. One sees that given parameters of our analysis,  the optimized operation strategy for type-1 is to send 2 replicas per generated data packet. Fig. \ref{oo1} represents the success probability for type-1 and type-2 traffic as a function of $n_1$ and $P_1$. One sees that success probability for type-1 increases to a point beyond which, the resulting interference from extra transmitted packets decline the improvement and deteriorate the performance. On the other hand, increase in the transmit power for type-1 devices,  increases the success probability for this type and severely decreases the performance of type-2 devices. It is also worthy to note that in Fig. \ref{oo1},  success probability increase in $n_i$ till $n_i=4$, however, from the battery lifetime analysis in Fig. \ref{oo2}, it is evident that battery lifetime decreases in $n_i$ for $n_i\ge 3$. This conclusion sheds light to the bound after which it is not feasible to trade battery lifetime for reliability.

  \subsubsection{Scalability Analysis}
Scalability analysis has been presented in Fig. \ref{sc}. Recall from section \ref{trf}, where we denoted the four degrees of freedom that can be leveraged to achieve a level of reliability in communications as i) transmit power, ii) number of replicas, iii) density of APs, and bandwidth of communications. First, Fig. \ref{scc} represents the way amount of provisioned network resources or devices resources must be scaled to comply with the increase in level of reliability in communications. It is clear that transmit power can be increase to a certain level in order to combat  noise. However, beyond a certain point increase in the transmit power cannot increase the success probability because it cannot combat interference. On the other hand, one sees that increase in number of replicas per packet can be used to increase reliability of communication. One must note that in scenarios with higher density of nodes, increasing number of replicas increases traffic load significantly and may even reduce reliability of communications. Finally, one sees the way increase in number of APs, and hence reducing the communications distance, or increasing the communications bandwidth, i.e. decreasing chance of collisions between nodes, can be used to achieve the required reliability level. 

Fig. \ref{sc2} represents the same results when $K=2$, i.e. two types of IoT devices are present in the service area, the required reliability level for type-1 traffic is 0.5, and $x$-axis represents density of type-2 devices (variable). Here, one sees that increasing transmit power in effective in combating both noise and interference\footnote{Recall the $\frac{P_k}{P_i}$ ratio in theorem \ref{t1}.}. It is also interesting to see that scaling bandwidth with the same rate as scaling the density of nodes can combat the extra interference due to the fact that the chance of collision scales down by scaling up the system bandwidth.

Fig. \ref{sc1} represent the scalability analysis for the case in which $K=1$, and $x$-axis represents density of nodes (variable). Similar to Fig. \ref{scc}, increasing transmit power cannot combat interference from the same type of nodes, and hence, it cannot be leveraged in adapting to the scaled network. On the other hand, increasing number of replica transmission can be useful to some extent because after some point, as we depicted in Fig. \ref{oo1}, increasing number of replicas  increases probability of collision significantly. 
\begin{figure}[t!]
    \centering
    \begin{subfigure}[t]{0.5\textwidth}
   \centering
 	\includegraphics[width=3.5in]{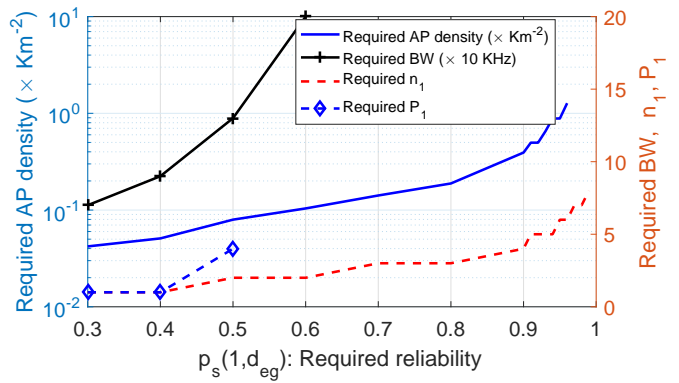}
 	\caption{Scalability analysis versus required reliability ($K=1$, $\lambda_1$=3.2).  }
 	\label{scc}
    \end{subfigure}%
\\
    \begin{subfigure}[t]{0.5\textwidth}
   \centering
 	\includegraphics[width=3.5in]{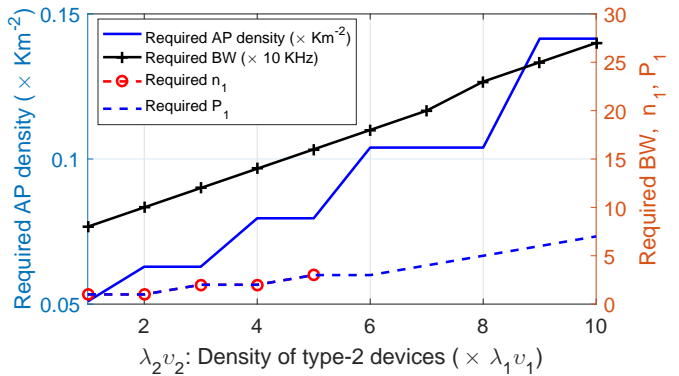}
 	\caption{Scalability analysis versus density of type-2 devices ($K=2$, $\lambda_1$=1.3, $P_1$=21 dBm, $P_2$=14 dBm, required $\text{p}_{\text s}(i,d_{\text{eg}})$=0.5).}
 	\label{sc2}
    \end{subfigure}%
\\
    \begin{subfigure}[t]{0.5\textwidth}
   \centering
 	\includegraphics[width=3.5in]{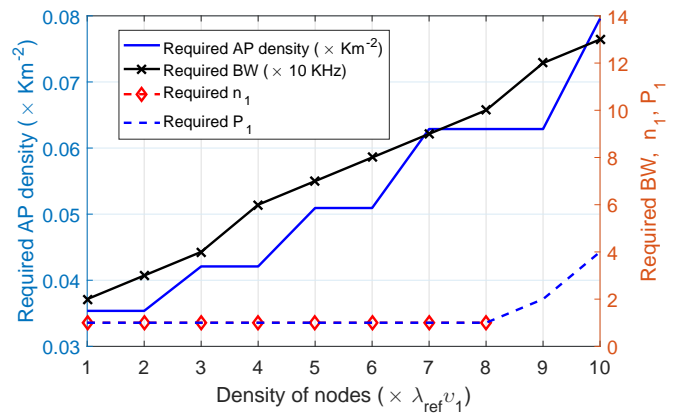}
 	\caption{Scalability analysis versus density of type-1 devices ($K=1$, $\lambda_{\text{ref}}$=0.3, required $\text{p}_{\text s}(i,d_{\text{eg}})$=0.5). Note: when the required value of one parameter for satisfying the reliability requirement is derived, the other parameters are set to the default value in Table \ref{sim}.}
 	\label{sc1}
    \end{subfigure}%
\caption{Scalability analysis}  \label{sc}
\end{figure}

 \section{Conclusion}
A tractable analytical model of reliability in  large-scale heterogeneous IoT networks has been presented. This model  has been employed subsequently in deriving tradeoff amongst cost of access network  and offered QoS to the  IoT traffic. The QoS for IoT traffic has been modeled in terms of reliability  and durability of communications. Using the derived results, optimized resource provisioning and operation control problems have been investigated, and optimized strategies aiming at lowering network cost while maximizing battery lifetime and complying with reliability constraints have been found. The derived expressions illustrate the way reliability in IoT connectivity can be  achieved by  sacrificing either battery lifetime or  sacrificing network cost, i.e. increasing amount of provisioned resources.
 The performance evaluation results  confirm existence of an optimal operation point before which, battery lifetime and reliability are both increasing in transmit power and number of replica transmission per packet; while beyond that point there is a tradeoff between reliability and battery lifetime. Furthermore, using the derived models the cost-optimized balance between provisioned radio and AP resources has been found. The accuracy yet tractability of the derived expressions in this work promotes use of them in network planning and optimization for future IoT networks.


\appendices

\section{Proof of Theorem \ref{t1}}\label{pt1}
When $m=1$, \eqref{suc} reduces to:
\begin{align}
&\hspace{2mm}\text{p}_{\text{s}}(i,{\bf z})= \exp\big(   \sum\limits _{j,k}-\lambda_k \int\nolimits_{\mathbb R^2} \big[1-\exp(-\hat\upsilon_{k,j} \mathcal U({\bf y}))\big]\bigg)\hspace{4mm}\nonumber\\
& \times \exp\big(-\frac{\mathcal N\gamma_{\text{th}}m}{\Omega  P_i \text{g}({\bf z})} \big) \times\int\nolimits_{\mathbb R^2} \exp\big(-{\textstyle \sum_j}\hat\upsilon_{i,j} \mathcal U({\bf y})\big)\text{f}({\bf y}) d {\bf y}
,\nonumber
\end{align}
where
$$\mathcal U({\bf y})\text{=}\int_{\mathbb R^2}\frac{\text{g}({\bf x}-{\bf y})}{ \text{g}({\bf x}-{\bf y})+ {\text{G}({\bf z})}}{\text{f}({\bf x}) d {\bf x}},$$
$$ \text{G}({\bf z})=\Omega  P_i \text{g}({\bf z})/[\gamma_{\text{th}} Q_jP_k].$$
Then, using Jensen's  inequality, we have:
\begin{align}
\text{p}_{\text{s}}(i,{\bf z}) \approx &   \text{P}_{{\text{\tiny N} }} \exp\big(  \sum\nolimits _{k,j}-\lambda_k\hat\upsilon_{k,j}\int\nolimits_{\mathbb R^2} \mathcal U({\bf y})\big)  d {\bf y} \nonumber\\
&\times\exp(-{ \sum\nolimits_{j}}\hat\upsilon_{i,j} \int\nolimits_{ \mathbb R^2} \mathcal U({\bf y})\text{f}({\bf y}) d {\bf y}),\nonumber\\
\approx &   \text{P}_{{\text{\tiny N} }} \exp\big(  \sum\limits _{k,j}\text{-}\lambda_k\hat\upsilon_{k,j}\int\nolimits_{\mathbb R^2} \frac{\text{g}({\bf y})}{\text{g}({\bf y})\text{+}\text{G}({\bf z})} d {\bf y}\int\nolimits_{\mathbb R^2}\text{f}({\bf v}) d {\bf v} \big)   \nonumber\\
&\times \exp\big(\text{-}{ \sum\limits_{j}}\hat\upsilon_{i,j} \int\limits_{\mathbb R^2}\int\limits_{\mathbb R^2} \frac{\text{g}({\bf x\text{-} y})}{\text{g}({\bf x\text {-}y})\text{+}\text{G}({\bf z})}\text{f}({\bf x}) d {\bf x}\text{f}({\bf y})   d {\bf y}\big)\nonumber\\
\approx &  \text{P}_{{\text{\tiny N} }}    \exp\big(-\sum\nolimits _{k,j}\lambda_k\hat\upsilon_{k,j}\int\nolimits_{\mathbb R^2} \frac{  \text{g}({\bf y})}{  \text{g}({\bf y})+  \text{G}({\bf z})} d {\bf y}\big)  \nonumber\\
&\times\exp\big(\text{-}{ \sum\limits_{j}}\hat\upsilon_{i,j}  \int\limits_{ \mathbb R^2} \int\limits_{ \mathbb R^2}\frac{  \text{g}({\bf v} )}{ \text{g}({\bf v} )\text{+}   \text{G}({\bf z})} \text{f}({\bf v}\text{+}{\bf y}) d {\bf v} \text{f}( {\bf y})  d {\bf y} \big),\nonumber
\end{align}
where $\bf v$ is an auxiliary variable equal to ${\bf x}-{\bf y}$. Using the isotropic property of  $\text{f}({\bf x})$, i.e.  $\text{f}({\bf y})=\text{f}(-{\bf y})$, and another change of variables, we have:
\begin{align}
\text{p}_{\text{s}}(i,{\bf z})\approx &  \text{P}_{{\text{\tiny N} }}    \exp\big(-\sum\nolimits _{k,j}\lambda_k\hat\upsilon_{k,j}\int\nolimits_{\mathbb R^2} \frac{  \text{g}({\bf y})}{  \text{g}({\bf y})+  \text{G}({\bf z})} d {\bf y}\big)  \nonumber\\
&\times\exp\big(\text{-}{ \sum\nolimits_{j}}\hat\upsilon_{i,j}  \int\nolimits_{ {\bf v}\in \mathbb R^2}\frac{  \text{g}({\bf v})}{ \text{g}({\bf v})\text{+}   \text{G}({\bf z})} \text{f}^*({\bf v}) d {\bf v} \big),\nonumber
\end{align}
where:
$$\int\nolimits_{{\bf y}\in \mathbb R^2} \text{f}({\bf v}\text{-}{\bf y})\text{f}( {\bf y}) d {\bf y}=\text{conv}(\text{f}({\bf v}),\text{f}({\bf v})) \buildrel \Delta \over = \text{f}^*({\bf v}).$$
By using $\text{H}(\cdot)$, as defined in \eqref{hf}, we have:
\begin{align}
\text{p}_{\text{s}}(i,{\bf z})&
\approx   \text{P}_{{\text{\tiny N} }} \big[\exp\big(-\sum\limits_{j\in\{1,2\}} \sum\limits_{k\in\mathcal K}\lambda_k\hat \upsilon_{k,j} \text{H}({\bf z},1, \frac{\gamma_{\text{th}}Q_j P_k}{\Omega  P_i})\big)\big]\nonumber\\
&\times\exp\big(-\sum\nolimits_{j\in\{1,2\}}{\hat\upsilon_{i,j}}  \text{H}({\bf z},\text{f}^*({\bf x}), \frac{Q_j\gamma_{\text{th}}}{\Omega })\big).\qed
\nonumber
\end{align}

\section{Proof of Corollary \ref{cne}}\label{pcne}

For $\text{f}({\bf x})$ given in \eqref{nor}, $\text{f}^*({\bf x})$ is derived as:
$$\text{f}^*({\bf x})={\exp(-||{\bf x}-{\bf x}_0||^2/(4\sigma^2))}/{ {4\pi\sigma^2}}.$$ Then, we have:
\begin{align}
\text{H}\big({\bf z},\text{f}^*({\bf x}), \xi)&=\frac{2\pi}{4\pi\sigma^2}\int\nolimits_{r=0}^{\infty}\frac{\exp(-r^2/(4\sigma^2))}{1+\frac{ r^4}{\xi||{\bf z}||^4}}  d {r}\nonumber\\
&=\frac{\xi||{\bf z}||^4}{2\sigma^2  }\int\nolimits_{r=0}^{\infty}\frac{\exp(-r^2/(4\sigma^2))}{[{\sqrt \xi||{\bf z}||^2}]^2+ { r^4} }  r d {r}\nonumber\\
&\buildrel (\text{e}) \over =\frac{\xi||{\bf z}||^4}{4   \sigma^2}\int\nolimits_{x=0}^{\infty}\frac{\exp(-a_1 x )}{(a_2)^2+ { x^2} }  {   d x}\nonumber\\
&\buildrel (\text{f}) \over = \frac{\sqrt\xi||{\bf z}||^2   }{4  \sigma^2  }\big[\text{ci}( \frac{\sqrt{\xi}||{\bf z}||^2}{4\sigma^2}  )\sin( \frac{\sqrt{\xi}||{\bf z}||^2}{4\sigma^2}  )\nonumber\\
&\hspace{2cm}-\text{si}( \frac{\sqrt{\xi}||{\bf z}||^2}{4\sigma^2}  )\cos(  \frac{\sqrt{\xi}||{\bf z}||^2}{4\sigma^2} )\big],\nonumber
\end{align}
where in (e) we have used change of variables $x=r^2$, (f) has been derived using table of integrals in  \cite[Eq.~3.352]{seri}, and $a_1=\frac{1}{4\sigma^2}$ and $a_2= {\xi||{\bf z}||^4}$ are auxiliary variables. 
\qed

\section{Proof of Theorem \ref{t3}}\label{pt3}
Using theorem \ref{t1}, corollary \ref{cne}, and corollary \ref{r1} we have:
\begin{align}
\text{p}_{\text{s}}(i,r)d \text{P}_{d_{\ell}}(r)
\approx &   \exp\big(\text{-}\frac{r^{4}\mathcal N\gamma_{\text{th}}}{\Omega  P_i \alpha} \big) \exp\big(\text{-}{\hat\upsilon_{i,2}} \big)\nonumber\\
 &\exp\big(\text{-}\sum\limits_{j,k} \lambda_k\hat \upsilon_{k,j} 
 r^2 \sqrt{\frac{\gamma_{\text{th}}Q_j P_k}{\Omega  P_i}} \frac{\pi^{2}}{2} \text{csc}(\frac{\pi}{2}) \big)\nonumber\\
 & \exp(\text{-}\lambda_{\text a}\pi r^2)\frac{2(\lambda_{\text a}\pi r^2)^{\ell}}{r({\ell}-1)!}dr\nonumber\\
 \approx &2rX_0 r^{2(\ell-1)}\exp(\text{-}[X_1r^{4}\text{+}X_2r^{2}])dr,\label{xs}
\end{align}
where $X_0$:$X_4$ have been defined in theorem \ref{t3}.
Inserting \eqref{xs} in \eqref{cov}, $\text{P}_\text{s}(i)$ is derived as:
\begin{align}
&1-\text{P}_\text{s}(i)\nonumber\\
&\approx  \prod\limits_{{\ell}} 1\text{-}\int\nolimits_{0}^{\infty} X_0 r^{2(\ell-1)}\exp(\text{-}[X_1r^{4}+X_2r^{2}])2rdr,\nonumber\\
 &\approx \prod\limits_{{\ell}}1\text{-} X_0 \int\nolimits_{0}^{\infty}  x^{(\ell-1)}\exp(-[X_1(x+\frac{X_2}{2X_1})^2\text{-}\frac{{X_2}^2}{4{X_1}^2} )])dx,\nonumber\\
 &\approx  \prod\limits_{{\ell}}1\text{-} X_0 \exp(\frac{{X_2}^2}{4{X_1}^2} )\int\nolimits_{\frac{{X_2}^2}{2X_1}}^{\infty}  (y-\frac{X_2}{2X_1})^{(\ell-1)}\exp(\text{-}X_1y^2)dy,\nonumber\\
  &\approx \prod\limits_{{\ell}}1\text{-} \frac{X_0}{\sqrt{{X_1}^{\ell-1}}} \exp(\frac{{X_2}^2}{4{X_1}^2} )\int\nolimits_{\frac{{X_2}^2}{2X_1}}^{\infty}  (z\text{-}\frac{X_2}{2\sqrt{X_1}})^{(\ell-1)}{\rm e}^{\text{-}z^2}dz,\nonumber\\
 &  \approx \prod\limits_{{\ell}} 1\text{-}\frac{X_0}{\sqrt{{X_1}^{\ell-1}}} \exp(\frac{{X_3}^2}{X_1} ) \mathcal G(X_3,\ell),
\end{align} 
 $x=r^2$, $y=x+\frac{X_2}{2X_1}$, and $z=\sqrt{X_1}y$ are auxiliary variables.
\qed

     \ifCLASSOPTIONcaptionsoff
  \newpage
\fi

\bibliographystyle{IEEEtran}
\bibliography{bibs}
 \end{document}